\documentclass{llncs}
\usepackage[OT4]{fontenc}
\usepackage{epsfig}
\usepackage{graphicx}
\usepackage{amsmath}
\usepackage{multicol}
\usepackage{xy}
\usepackage[utf8]{inputenc}
\usepackage{algorithm}
\usepackage[noend]{algpseudocode}
\usepackage{enumitem}
\title{Faster Algorithms for Markov Decision Processes with Low Treewidth}

\author{Krishnendu Chatterjee\inst{1} and Jakub Łącki\inst{2}}

\institute{
IST Austria (Institute of Science and Technology Austria)
    \and
Institute of Informatics, University of Warsaw, Poland 
}

\newcommand{\asr}{ASR }
\newcommand{\mec}{MEC }
\newcommand{\kexp}{2.38}

\begin{document}
\maketitle

\begin{abstract}
We consider two core algorithmic problems for probabilistic verification:
the maximal end-component decomposition and the almost-sure reachability 
set computation for Markov decision processes (MDPs). 
For MDPs with treewidth $k$, we present two improved 
static algorithms for both the problems that run in time 
$O(n \cdot k^{\kexp} \cdot 2^k)$ and $O(m \cdot \log n \cdot k)$,
respectively, where $n$ is the number of states and $m$ is the number of edges, 
significantly improving the previous known $O(n\cdot k \cdot \sqrt{n\cdot k})$ 
bound for low treewidth.
We also present decremental algorithms for both problems for MDPs with constant treewidth that run in amortized logarithmic time,
which is a huge improvement over the previously known algorithms
that require amortized linear time.  
\end{abstract}

\section{Introduction}
In this work we will present efficient static and decremental algorithms for 
two core graph algorithmic problems in probabilistic verification when the 
graph has low treewidth. 
We start with the basic description of the model, the problem, and its importance.

\smallskip\noindent{\bf Markov decision processes with parity objectives.} 
The standard model of systems in probabilistic verification that exhibit 
both probabilistic and nondeterministic behavior are \emph{Markov decision processes (MDPs)}~\cite{Howard}.
MDPs have been used for control problems for stochastic systems~\cite{FV97},
where nondeterminism represents the freedom of the controller to choose a 
control action, and the probabilistic component of the behavior describes the 
system response to control actions; as well as in many other applications~\cite{CY95,BdA95,PRISM}.
A \emph{specification} describes the set of good behaviors of the system.
In the verification and control of stochastic systems the specification is 
typically an $\omega$-regular set of paths.
The class of $\omega$-regular languages extends classical regular languages to 
infinite strings, and provides a robust specification language to express
all commonly used specifications, such as safety, liveness, fairness, etc.~\cite{Thomas97}. 
A canonical way to define such $\omega$-regular specifications are \emph{parity} objectives.
Hence MDPs with parity objectives provide the mathematical framework to 
study problems such as the verification and control of stochastic systems.

\smallskip\noindent{\bf The analysis problems.} 
There are two types of analysis for MDPs with parity objectives.
The \emph{qualitative analysis} problem given an MDP with a parity 
objective, asks for the computation of the set of states from
where the parity objective can be ensured with probability~1 
(almost-sure winning).
The more general \emph{quantitative analysis} asks for the computation 
of the maximal probability at each state with which the controller can 
satisfy the parity objective. 

\smallskip\noindent{\bf Significance of qualitative analysis.} 
The qualitative analysis of MDPs is an important problem in verification. 
In several applications the controller must ensure that the correct 
behavior arises with probability~1.
For example, in analysis of randomized embedded schedulers, the relevant 
questions is whether every thread progresses with probability~1~\cite{EMSOFT05}.
Moreover, even in applications where it is sufficient to satisfy the 
specification with probability $p<1$, the correct choice of $p$ is a 
challenging problem, due to the simplifications introduced during modeling; 
for example, for randomized distributed algorithms it is common to require 
correctness with probability~1 (see, e.g., \cite{PSL00,KNP_PRISM00,Sto02b}). 
Furthermore, in contrast to quantitative analysis, 
qualitative analysis is robust to numerical perturbations and 
precise transition probabilities, and consequently the algorithms for 
qualitative analysis are discrete and combinatorial.
Finally, the best known algorithms for quantitative analysis of 
MDPs with parity objectives first perform the qualitative analysis, and then 
a quantitative analysis on the result of the qualitative analysis~\cite{CY95,luca-thesis,CJH04}.

\smallskip\noindent{\bf Core algorithmic problems.} 
The qualitative analysis of MDPs with parity objectives relies on two graph
algorithmic problems: (1)~the maximal end-component decomposition; and 
(2)~the almost-sure reachability set computation. 
An end-component $C$ in an MDP is a set of states that is strongly connected 
and closed (no probabilistic transition from $C$ leaves $C$), and a maximal
end-component is an end-component which is maximal with respect to inclusion 
ordering. 
The maximal end-component (MEC) problem generalizes the scc (maximal strongly 
connected component) decomposition problem for directed graphs, and recurrent 
classes for Markov chains. 
The almost-sure reachability set for a set $U$ of target vertices is the set of
states such that it can be ensured that the set $U$ is reached with probability~1
(in other words, it is the qualitative analysis for reachability objectives).
The qualitative analysis problem for MDPs with parity objectives with $d$-priorities
can be solved with $\log d$ calls to the \mec decomposition problem and one call
to the almost-sure reachability problem~\cite{CH11}.
Thus the \mec decomposition and the almost-sure reachability set computation are
the core algorithmic problems required for the qualitative analysis of MDPs with
parity objectives.
In addition to qualitative analysis of MDPs with parity objectives, 
several algorithms for quantitative analysis of MDPs with 
quantitative objectives such as $\limsup$ and $\liminf$ 
objectives~\cite{CH-ILC}, combination of mean-payoff and 
parity objectives~\cite{CHJS10}, and multi-objective mean-payoff objectives~\cite{BBCFK11}, 
rely crucially on the \mec decomposition problem.

\smallskip\noindent{\bf Dynamic algorithms.}
In the design and analysis of probabilistic systems it is 
natural that the systems under verification are developed incrementally 
by adding choices or removing choices for player~1, whereas the probabilistic
choices which represent choice of nature or uncertainty remain unchanged. 
Hence there is a clear motivation to obtain dynamic algorithms for \mec 
decomposition and almost-sure reachability set for MDPs that achieve a 
better running time than recomputation from scratch when player-1 edges 
are inserted or deleted.

\smallskip\noindent{\bf Previous results.} 
The current best known algorithms for both the \mec decomposition and the 
almost-sure reachability set computation require 
$O(m \cdot \min(\sqrt{m},n^{2/3}))$ time~\cite{CH11,CH12},
where $n$ is the number of states and $m$ is the number of transitions 
(edges).
Using a well-known fact that graphs of treewidth $k$ have $O(n\cdot k)$ edges,
one can obtain $O(n\cdot k \cdot \sqrt{n\cdot k})$ algorithms for \mec decomposition and almost-sure reachability set computation (they follow directly from the general $O(m \cdot \sqrt{m})$-time algorithm). The best known incremental 
and decremental algorithms for both problems require amortized linear time
($O(n)$ time)~\cite{CH11}.

\smallskip\noindent{\bf Our contributions.} In this work we consider MDPs 
with low treewidth.
The concept of treewidth and tree decomposition of graphs was introduced
in~\cite{RS84}.
On one hand treewidth is a very relevant graph theoretic notion that 
measures how a graph can be decomposed into a tree, on the other hand,
most systems developed in practice have low treewidth.
For example, it has been shown that the control flow graphs of goto free Pascal 
programs have treewidth at most~3, and that the control flow graphs of goto 
free C programs have treewidth at most~6~\cite{Tho98}.
It was also shown in~\cite{Tho98}  that tree decompositions, which are
very costly to compute in general, can be generated in linear time with small
constants for these control flow graphs.
Our main results are efficient static and decremental algorithms for the
\mec decomposition and the almost-sure reachability set computation for MDPs 
with low treewidth.
Several benchmarks in PRISM are probabilistic programs written in programming 
languages mentioned above and consequently have small treewidth, and our results
are relevant for such MDPs.
The details of our contribution are as follows:
\begin{enumerate}
\item We present two improved static algorithms both for the \mec decomposition 
and the almost-sure reachability set computation for MDPs with treewidth $k$ 
that run in time $O(n \cdot k^{\kexp} \cdot 2^k)$ and $O(m \cdot \log n \cdot k)$,
respectively, where $n$ is the number of states and $m$ is the number of edges
(also note that for treewidth $k$ we have $m=O(n \cdot k)$).
For MDPs with low treewidth, our new linear-time algorithms are significant 
improvements over the previous known $O(n\cdot k \cdot \sqrt{n \cdot k})$ 
algorithms for both the problems.

\item We present decremental algorithms for the \mec decomposition and the
almost-sure reachability set computation for MDPs with treewidth $k$ that
require $O(k \cdot \log n)$ amortized time, which is a huge improvement
for constant treewidth over the previous algorithms that require $O(n)$ amortized time. 
\end{enumerate}
Our key technical contribution is as follows: for MDPs we establish a separation
property for the almost-sure reachability set that allows us to use tree
decomposition to obtain the  $O(n \cdot k^{\kexp} \cdot 2^k)$-time static algorithm. A 
similar intuition also works for the \mec decomposition problem.
We then view the \mec decomposition and the almost-sure reachability set computation
problems as decremental graph problems, and use dynamic graph algorithmic techniques
to obtain the $O(m \cdot \log n \cdot k)$-time static algorithms and the decremental
algorithms.
Note that when edges are inserted, the treewidth of the graph may increase and the tree decomposition can change. 
Thus, incremental algorithms with polylogarithmic amortized cost
remain an interesting open question (even for scc decomposition).

\smallskip\noindent{\bf Related works.} The notion of treewidth is studied in 
context of many graph theoretic algorithms, see~\cite{Bo97} for an excellent survey.
In verification, the problem of low and medium treewidth has been considered
for efficient algorithms for parity games: a polynomial time algorithm for
parity games with constant treewidth was presented in~\cite{Ob03}; a recent improved 
result for constant treewidth was presented in~\cite{FS12}; and the algorithmic 
problem of parity games with medium treewidth was considered in~\cite{FL11}.
Though the games problem has been studied with the treewidth restriction,
to the best of our knowledge, improved algorithms for MDPs have not been 
considered with the treewidth restriction.

\section{Preliminaries}
In this section we first present the basic graph theoretic 
definitions of the \mec decomposition and the almost-sure reachability
set computation, and then define the notions of treewidth.

\subsection{\mec decomposition and almost-sure reachability}

\noindent{\em Markov decision processes (MDPs).}
A \emph{Markov decision process (MDP)} $G=((V, E), (V_1,V_P),\delta)$ 
consists of a finite directed {\em MDP graph} $(V,E)$, a partition $(V_1,V_P)$ 
of the \emph{finite} set $V$ of vertices, and a probabilistic transition function 
$\delta$: $V_P \rightarrow {\mathcal D}(V)$, where ${\mathcal D}(V)$ denotes the 
set of probability distributions over the vertex set $V$, such that for all vertices
$u \in V_P$ and $v \in V$  we have $uv \in E$ iff $\delta(u)(v)>0$.
An edge $uv \in E$ is a \emph{player-1} edge if $u \in V_1$.
For the algorithmic problems we will consider, the probabilistic transition 
function will not be relevant and we will consider the MDP graph along with the 
partition.

\paragraph{Maximal end-component decomposition.}
For the maximal end-component decomposition, the input is a directed graph $G = (V,E)$ 
and a partition $(V_1, V_P)$ of its vertex set (i.e., the MDP graph and the partition).
An end-component $U$ is a set of vertices such that the subgraph induced by $U$ is strongly connected and 
for each edge $uv \in E$, if $u \in U \cap V_P$ then $v \in U$.
If $U_1$ and $U_2$ are two end-components and $U_1 \cap U_2 \neq \emptyset$, 
then $U_1 \cup U_2$ is also an end-component.
The maximal end-component (MEC) decomposition consists of all the maximal end-components of $V$ and all vertices of $V$ 
that do not belong to any MEC.

\paragraph{Almost-sure reachability.} 
For almost-sure reachability, the input is an MDP and a target set $U \subseteq V$ 
of vertices, and the goal is to compute the set $A$ of vertices,
such that player~1 can ensure that the set $U$ is reached with probability~1.
We first note that given the target set $U$, we can add a new vertex $s$ as the new
target vertex, and transform the set $U$ such that all out-edges from vertices in $U$ 
end up in $s$, and the vertex $s$ has only a self-loop.
Thus we will consider the case when the target set is a single vertex $s$.
We first reduce the computation of the almost-sure reachability set for 
a target vertex $s$ to the following problem.
The input is a directed graph $G = (V,E)$, a partition $(V_1, V_P)$ of its vertex set (the MDP graph and
the partition),  and a target vertex $s \in V$.
The goal is to compute a maximal (w.r.t inclusion) subset $Q \subseteq V$, such that
the following two conditions are satisfied:
\begin{itemize}
\item for every $q \in Q$, there exists a path from $q$ to $s$ consisting only of vertices in $Q$ (global condition), and
\item for every $uv \in E$, if $u \in Q \cap V_P$, then $v \in Q$ (local condition).
\end{itemize}
First observe that if $Q_1 \subseteq V$ and $Q_2 \subseteq V$ both satisfy the global 
and the local conditions, then so does $Q_1 \cup Q_2$.
It follows that there is a unique maximum set $A^* \subseteq V$ that satisfies
both the global and the local conditions.
The resulting set $A^*$ is the almost-sure reachability set (in the following also called an \emph{\asr set}). 
Let $A$ be the almost-sure reachability set and $A^*$ be the largest set that 
satisfies the two conditions (the global and the local conditions).
\begin{lemma}\label{lemm:equalA}
We have $A=A^*$.
\end{lemma} 
Since $A=A^*$ we consider the graph theoretic problem of computation of $A^*$ 
(i.e., the largest set satisfying the global and the local conditions).

\smallskip\noindent{\em Notations.} 
Let $G$ be a directed graph.
We denote its vertex and edge set by $V(G)$ and $E(G)$, respectively.
By $G[S]$ we denote the subgraph of $G$ induced on vertices belonging to $S$, whereas by $G \setminus S$ we denote the subgraph of $G$ induced on $V(G) \setminus S$.
A \emph{separator} is a subset $S \subseteq V(G)$, such that $G \setminus S$ has more connected components than $G$ (when all edges are treated as undirected).

\subsection{Tree decomposition of graphs}
\begin{figure}
\begin{centering}
\includegraphics[scale=0.2]{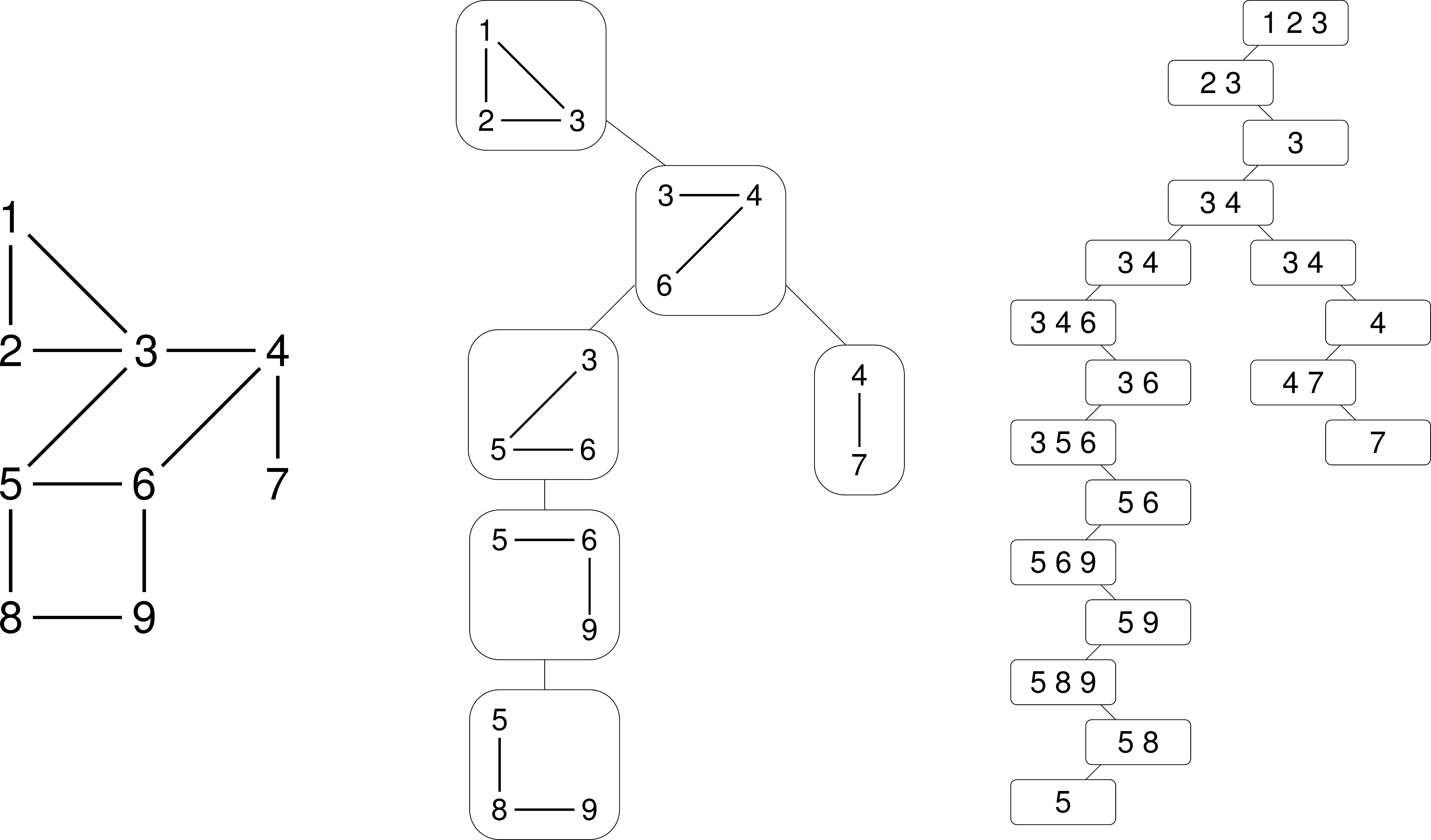}
\caption{\label{fig:tree-dec}A sample graph (left), its tree decomposition (center, edges covered by each bag have been marked for illustration) and a nice tree decomposition (right).}
\end{centering}
\end{figure}

We begin by introducing some definitions depicted in Fig.~\ref{fig:tree-dec}.
\begin{definition}
Let $G=(V,E)$ be an undirected graph.
A \emph{tree decomposition} of $G$ is a pair $(B, T)$, where $B$ is a family 
$B_1, \ldots, B_n$ of subsets of $V$ (called bags) and $T$ is a tree, whose 
nodes are sets $B_i$.
The decomposition satisfies the following properties:
\begin{enumerate}
\item $\bigcup B_i = V$ (bags cover vertices).
\item For every $uv \in E$ there exists $B_j$, such that $u, v \in B_j$ (bags cover edges).
\item For every $v \in V$ the sets $B_i$ containing $v$ form a connected subtree of $T$.
\end{enumerate}
\end{definition}

\begin{definition}
The \emph{width} of a tree decomposition $(B,T)$ is equal to $\max_{B_i \in B} |B_i| - 1$.
The treewidth of an undirected graph is the minimal possible width of its tree decomposition.
\end{definition}

The concept of treewidth grasps the \emph{sparseness} of a graph.
Treewidth of a tree is equal to $1$, while cliques on $n$ vertices have treewitdth $n-1$.
Note that the definitions are given for undirected graphs, but they can also be applied to directed graphs.
In such case, we treat all edges as undirected.

\begin{definition}
A tree decomposition $(B, T)$ is called \emph{nice} if $T$ is a rooted tree and each of its nodes $B_i$ belongs to one of the following four types:
\begin{enumerate}
\item \textbf{leaf} --- $B_i$ is a leaf of $T$ and $|B_i|=1$.
\item \textbf{introduce} --- $B_i$ has a single child $B_j$ and $B_i = B_j \cup \{v\}$.
\item \textbf{forget} --- $B_i$ has a single child $B_j$ and $B_i = B_j \setminus\{v\}$.
\item \textbf{join} --- $B_i$ has two children $B_j$ and $B_k$, and $B_i = B_j = B_k$.
\end{enumerate}
\end{definition}

\begin{theorem}[\cite{Bodlaender96}]\label{thm:bod}
Let $G$ be a graph of treewidth $k$.
Assuming that $k$ is a constant, the tree decomposition of $G$ of width $k$ can be computed in $O(n)$ time. 
\end{theorem}

\begin{lemma}[see e.g.~\cite{Kloks94}]
\label{lemma:nice_decomposition}
A tree decomposition can be transformed, in linear time, into a nice tree decomposition of the same width, consisting of $O(n)$ nodes.
\end{lemma}

We also use the following well-known fact, which can be derived from the definition.
Consider a vertex $t_B$ of a tree decomposition $T$ of a graph $G$ and assume that it contains a bag $B \subset V(G)$.
Denote the connected components of $T \setminus \{t_B\}$ by $T_1, \ldots, T_k$.
Then, each connected component of $G \setminus B$ corresponds to one of the subtrees $T_i$ (i.e. the vertices from the connected component are all covered only by bags from the subtree $T_i$).
This is formalized in the Lemma below.

\begin{lemma}
\label{lemma:separator}
Let $B$ be a bag in a node $t_B$ of the tree decomposition of $G$.
Consider the connected components $T_1, \ldots, T_k$ of $T \setminus \{t_B\}$.
Then the following hold:
\begin{enumerate}
\item Either $B$ is a separator in $G$ or all but one $T_i$ consist solely of bags that are subsets of $B$.

\item Each path from a vertex $u \not\in B$ covered with a bag in $T_i$ to a vertex $v \not\in B$ covered with a bag in $T_j$ ($i \neq j$) goes through a vertex in $B$.
\end{enumerate}
\end{lemma}

Observe that a vertex not belonging to $B$ can be covered by bags from at most one $T_i$.
This is because the set of bags covering a given vertex forms a connected subgraph of $T$.

\section{Algorithms for MDPs with Constant Tree-width}
In this section we will first present an algorithm for computing the \asr set, 
whose running time depends linearly on the 
size of the input MDP graph, where the input graph has constant treewidth.
We will then present the linear-time algorithm for \mec decomposition for MDPs with constant treewidth graphs.
The algorithms require that a tree decomposition of the graph of width $k$ 
is given and run in time that is exponential in $k$.
If $k$ is a constant, the decomposition can be computed in linear time (see Theorem~\ref{thm:bod}).
To simplify presentation, we use Lemma~\ref{lemma:nice_decomposition} to transform the decomposition 
to a \emph{nice} one.

\subsection{Almost-sure reachability}
Our algorithm for the \asr set computation is based on the following separation property.

\begin{lemma}
\label{lemma:separation_property}
Let $B$ be a subset of $V(G)$, such that the target vertex $s$ belongs to $B$.
Denote the connected components of $G \setminus B$ by $C_1, \ldots, C_k$.
Assume that we know the intersection of the \asr set $A$ with $B$.
For each $i = 1, \ldots, k$, construct the subgraph of $G$ induced on $C_i \cup B$.
Add to this graph a set of edges $\{vs | v \in A \cap B\}$, thus obtaining a \emph{patched} component $\overline{C}_i$.
Denote by $A_i$ the \asr set in $\overline{C}_i$.
Then we have $A = A_1 \cup \ldots \cup A_k$.
\end{lemma}

Lemma~\ref{lemma:separation_property} says that if we know $A \cap B$, then we can compute the \asr set 
independently in each (patched) connected component of $G \setminus B$ and then simply merge the results.
Since we assume that $G$ has low treewidth, it also has separators of small size.
Thus, in the algorithm we can guess $A \cap B$, by checking all possibilities.
We do not prove the separation property explicitly.
Instead, we give the algorithm inspired with this property and then prove its correctness.
The property will follow from Lemma~\ref{lemma:p_is_maximal_asr}.
Let us now describe the details.

Denote the nice tree decomposition of $G$ by $T$.
Without loss of generality we may assume that the target vertex $s$ belongs to every bag of $T$ and that the decomposition is rooted in a node with bag $\{s\}$.
Indeed, if this is not the case, we may modify the decomposition $T$ as follows.
First, we add $s$ to every bag.
Then, for every leaf of $T$ that contains two vertices in its bag, we add a child with a bag $\{s\}$.
If after the two steps we have a node $d$ with a single child $c$ such that the bags of $d$ and $c$ are equal, we merge them together (i.e. contract the edge connecting them).
Lastly, we add a new root node with a bag $\{s\}$ and connect it with a chain of \textbf{forget} nodes to the original root.
It is easy to see that the process increases the width of $T$ at most by one and yields a nice tree decomposition.

The algorithm is based on a bottom-up dynamic programming on $T$.
Fix a node $d$ of $T$, and assume that it contains a bag $B_d$.
Denote by $G_d$ the subgraph of $G$ induced on the vertices enclosed in the bags from the subtree rooted at $d$.
By Lemma~\ref{lemma:separator}, $B_d$ separates $G_d \setminus B_d$ from the rest of the graph.

Now, according to Lemma~\ref{lemma:separation_property}, for each subset $B' \subseteq B_d$ we should add edges $\{vs | v \in B'\}$ to $G_d$ and compute the \asr set of the obtained graph.
However, we do a slightly different thing: instead of adding edges, we just treat all vertices of $B'$ as target vertices (note that this has the same effect
as adding edges from vertices in $B'$ to $s$).
This motivates the following definition of a partial solution.
Partial solution is defined with respect to a subgraph of $G_d \subseteq G$, and, informally, it is the set of vertices from $G_d$ that will be included in the \asr set.

\begin{definition}
\label{definition-partial-solution}
A \emph{partial solution} for a node $d$ is a subset of $V(G_d)$. A partial solution $P$ is called \emph{valid}, if the following hold. 
\begin{enumerate}[label=\roman{*}., ref=(\roman{*})]
\item\label{vps:first} For every $v \in P \cap V_P$ and every edge $vu \in E(G_d)$, we have $u \in P$.
\item\label{vps:second} For every $v \in P$ there exists a path in $P$ that connects $v$ to some vertex in $P \cap B_d$.
\end{enumerate}
\end{definition}

We denote by $P(B', d)$ the maximal (w.r.t. inclusion) valid partial solution (for node $d$) which satisfies $P(B', d) \cap B_d = B'$.\footnote{In the end we prove slightly less about the values $P(\cdot, \cdot)$ that are computed by the algorithm, but it is convenient to think about them this way.}
Observe that the definition is unambiguous, since the union of two valid partial solutions is a valid partial solution.
However, it might be the case that for some choice of $B'$ there are no feasible valid partial solutions.
In such a case we set $P(B', d) = \bot$.
We later show that if $B' = A \cap B_d$, then $P(B', d) = A \cap V(G_d)$.

The algorithm considers possible ways of including a subset of $B_d$ in the \asr set, by iterating through all \emph{valid} subsets $B' \subseteq B_d$.
A subset $B' \subseteq B_d$ is valid, if it contains the target $s$ and for each $v \in B' \cap V_P$ and every edge $vu \in E \cap (B_d \times B_d)$, we have $u \in B'$.
In particular, for any valid partial solution $P$ containing $s$, the set 
$P \cap B_d$ is a valid subset.

In addition to $P(B', d)$, for each valid $B' \subseteq B$ and each pair of vertices $x, y \in B'$, we compute whether there exists an $x$-to-$y$ path consisting of vertices contained in $P(B', d)$.
Formally, we compute the transitive closure of $G[P(B', d)]$, restricted to $B'$.
In the following this transitive closure is denoted by $TC(B', d)$.
Note that it is a subset of $B_d \times B_d$.

The algorithm is run bottom-up on $T$.
For a given node $d$ and each valid subset $B'$ it computes $P(B', d)$ and $TC(B', d)$, using the values from the children of $d$.
There are four cases to consider, one for each type of node.
In the description, we assume that the value $\bot$ is \emph{propagating}.
This means, that the result of any set operation involving $\bot$ is $\bot$.

\begin{itemize}
\item \textbf{Leaf} The bag contains a single vertex $s$ (the target), so the transitive closure is empty and we set $P(\{s\}, d) = \{s\}$.
\item \textbf{Join} Denote the children of $d$ by $c_1$ and $c_2$.
In this case, we set $P(B', d) = P(B', c_1) \cup P(B', c_2)$, so the transitive closures from the children have to be combined, i.e. $TC(B', d) = (TC(B', c_1) \cup TC(B', c_2))^{*}$.
The asterisk denotes the operation of computing the transitive closure.

\item \textbf{Introduce} Denote the introduced vertex by $w$ and the child of $d$ by $c$.
For all valid subsets $B' \subseteq B_d$ that do not contain $w$, we set $P(B', d) = P(B', c)$ and $TC(B', d) = TC(B', c)$.
If $w \in B'$, then $P(B', d) = P(B' \setminus\{w\}, c) \cup \{w \}$.
Thus, to compute the transitive closure in this case, we take $TC(B' \setminus \{w\}, c)$, add all edges incident to $w$ and compute the transitive closure of the obtained set.
Hence, $TC(B', d) = (TC(B' \setminus \{w\}, c) \cup \{wz \in E(G) | z \in B'\} \cup \{zw \in E(G) | z \in B'\})^{*}$.

\item \textbf{Forget} Denote the vertex that is forgotten by $w$ and the child of $d$ by $c$.
Hence, the bag in the child $B_c$ is equal to $B_d \cup \{w\}$.
We check whether we can include $w$ in $P(B', d)$.
For this, condition~\ref{vps:second} (of Definition~\ref{definition-partial-solution}) has to hold, i.e., there has to be a path in $P(B', d)$ that connects $w$ to some vertex in $B'$.
We claim that it suffices to check, whether $w$ has any out-edges in $TC(B' \cup \{w\}, c)$.
If this is the case, then $w$ is connected to some vertex from $B'$ in $P(B' \cup \{w\}, c)$, so $P(B', d) = P(B' \cup \{w\}, c)$ and we can set $TC(B', d) = TC(B' \cup \{w\}, c) \cap (B' \times B')$.
Otherwise, we just copy the result from the child, that is set $P(B', d) = P(B', c)$ and $TC(B', d) = TC(B', c)$.
\end{itemize}

Finally, the \asr set computed by the algorithm is stored in $P(\{s\}, r)$.
We now prove the correctness of the algorithm with the following two lemmas.
The proof of Lemma~\ref{lemma:p_is_vps} is presented in the appendix.




\begin{lemma}
\label{lemma:p_is_vps}
For each node $d$ and each valid subset $B' \subseteq B_d$, if $P(B', d) \neq \bot$, then $P(B', d)$ is a valid partial solution and $TC(B', d)$ is computed correctly.
\end{lemma}

\begin{lemma}
\label{lemma:p_is_maximal_asr}
Let $A$ be the maximum \asr set.
For each node $d$, $P(A \cap B_d, d) = A \cap V(G_d)$.
\end{lemma}
\begin{proof}
The proof proceeds by induction on the depth of the subtree rooted in $d$.
First, it is easy to see that $A \cap B_d$ is a valid subset for $d$.
Moreover, $A \cap V(G_d)$ is a valid partial solution for $d$.
Let us check condition~\ref{vps:second} of Definition~\ref{definition-partial-solution}.
For each $v \in A$ there exists an $v$-to-$s$ path $p$ in $A$.
Denote by $v_l$ the last vertex of $p$ that lies inside $A \cap V(G_d)$.
By Lemma~\ref{lemma:separator}, $v_l \in B_d$ and consequently also $v_l \in A \cap B_d$.

\begin{itemize}
\item \textbf{Leaf} $P(A \cap B_d, d) = P(\{s\}, d) = \{s\} = A \cap V(G_d)$.
\item \textbf{Join} By induction hypothesis we have $P(A \cap B_{c_i}, c_i) = A \cap V(G_{c_i})$, for $i = 1, 2$.
From the definition $P(A \cap B_d, d) = P(A \cap B_d, c_1) \cup P(A \cap B_d, c_2) = P(A \cap B_{c_1}, c_1) \cup P(A \cap B_{c_2}, c_2) = (A \cap V(G_{c_1})) \cup (A \cap V(G_{c_2})) = A \cap (V(G_{c_1}) \cup V(G_{c_2})) = A \cap V(G_d)$.
\item \textbf{Introduce} If $A$ does not contain the introduced vertex $w$, then $P(A \cap B_d, d) = P(A \cap B_c, c) = A \cap V(G_c) = A \cap (V(G_d) \setminus \{w\}) = A \cap V(G_d)$.
Otherwise, if $w \in A$ we have $P(A \cap B_d, d) = P((A \cap B_d) \setminus\{w\}, c) \cup \{w\} = (A \cap V(G_c)) \cup \{w\} = A \cap (V(G_d) \setminus \{w\}) \cup \{w\} = A \cap V(G_d)$.
\item \textbf{Forget} 
Denote the forgotten vertex by $w$.

We claim that $w \in A$ iff $A \cap B_c$ is a valid subset of $B_c$ and $w$ has some out-edges in $TC(A \cap B_c, c)$.
($\Rightarrow$) It follows immediately that $A \cap B_c$ is a valid subset.
Moreover, since there is a path from $w$ to $s$ in $A$, by Lemma~\ref{lemma:separator}, there has to be a path that connects $w$ to some vertex in $(A \cap B_c) \setminus\{w\}$ in $P(A \cap B_c, c)$.
($\Leftarrow$) Assume that $w \not\in A$.
We show that $A \cup P((A \cap B_c) \cup \{w\}, c)$ is an almost-sure reachable set that is larger than $A$.
Indeed, we know that from every vertex in $P((A \cap B_c) \cup \{w\}, c)$ there is a path to a vertex in $(A \cap B_c) \cup \{w\}$, hence also a path to $A \cap B_c$.
In addition, from every vertex in $A \cap B_c$ there is a path to $s$.
It follows easily that condition~\ref{vps:second} of being an \asr set also holds, which shows the desired.

Now, if $w \in A$, the algorithm sets $P(A \cap B_d, d) = P((A \cap B_d) \cup \{w\}, c) = P(A \cap B_c, c) = A \cap V(G_c) = A \cap V(G_d)$.
On the other hand, if $w \not\in A$, we have $P(A \cap B_d, d) = P(A \cap B_d, c) = P(A \cap (B_c \setminus \{w\}), c) = P(A \cap B_c, c) = A \cap V(G_c) = A \cap V(G_d)$. 
\qed
\end{itemize}
\end{proof}

By applying Lemma~\ref{lemma:p_is_maximal_asr} to the root $r$ of the tree decomposition, we obtain that $P(A \cap V(G), r) = A \cap V(G) = A$.
Let us now analyze  the running time. 

\smallskip\noindent{\em Running time analysis.}
We represent $TC(\cdot, \cdot)$ with a $(k+2) \times (k+2)$ matrix.
(In the original tree decomposition bags had size $k+1$, but then we added the  
vertex $s$ to every bag.)
The sets $P(\cdot, \cdot)$ can be represented implicitly, that is for a set $P(B, d)$ we store how it can be obtained from the respective sets contained in the children of $d$.
This requires constant memory for each set.
We iterate through $O(2^k)$ subsets of each bag.
Checking whether a set is valid boils down to inspecting all edges inside a bag, which can be done in $O(k^2)$ time.
The most costly operation performed for each valid subset is the computation of the transitive closure of a graph containing $O(k)$ vertices.
This can be achieved in $O(k^{\kexp})$ time by using fast matrix multiplication (\cite{Coppersmith90}, \cite{Williams12}).\footnote{In practice, a simple $k^3$ algorithm might a better choice than algebraic algorithms for multiplying matrices.}
Restoring the result takes time that is linear in the size of the tree decomposition.
By Lemma~\ref{lemma:nice_decomposition}, the decomposition consists of $O(n)$ nodes.
Hence, the algorithm runs in $O(n \cdot 2^k \cdot k^{\kexp})$ time.
We obtain the following result.

\begin{theorem}
Given an MDP and its tree decomposition of width $k$ of the MDP graph, the \asr set can be computed in $O(n\cdot 2^k\cdot k^{\kexp})$ time, where $n$ is the number of states (vertices).
\end{theorem}

\subsection{\mec decomposition}
The algorithm is similar to the one for the \asr set in that it is also based on dynamic programming on a tree decomposition.
Again, we assume that we have a nice tree decomposition with a bag of size $1$ in the root.
This time we obviously do not add the target vertex to every bag, as there is no distinguished vertex.

As in the previous algorithm, we define a partial solution for a node $d$ to be a subset of $V(G_d)$.
This subset consists of vertices that are to form a single MEC.
A partial solution $P$ is valid, if three conditions hold.

\begin{enumerate}
\item For every $v \in P \cap V_P$ and every edge $vu \in E(G_d)$, we have $u \in P$.
\item For every $v \in P$ there exists a path in $P$ from $v$ to some vertex in $P \cap B_d$.
\item For every $v \in P$ there exists a path in $P$ from some vertex in $P \cap B_d$ to $v$.
\end{enumerate}

Note that the only difference from the algorithm for \asr set is that we have added the third condition.
As a result we can use the dynamic programming scheme from the previous section, with only a slight change.
When we perform a check that depends on the second condition (while processing a \textbf{forget} node), we need to run two symmetric checks instead of one.
Let $P(B', d)$ denote the maximal partial solution for $d$ such that $P(B', d) \cap B_d = B'$.

We use the following two lemmas to show the correctness of the algorithm.
Their proofs can be obtained easily from the proofs of their analogous lemmas in the previous section.

\begin{lemma}
For each node $d$ and each valid subset $B' \subseteq B_d$, $P(B', d)$ is a valid partial solution and $TC(B', d)$ is computed correctly.
\end{lemma}

\begin{lemma}
\label{lemma:mec_intersection}
For every node $d$ and \mec $M$ such that $M \cap B_d \neq \emptyset$, we have $P(M \cap B_d, d) = M \cap V(G_d)$.
\end{lemma}

The difference in this algorithm is in obtaining the result after the dynamic programming step is finished.
First, we find the rootmost (that is, the one closest to the root) node $d_1$ and a vertex $v_1 \in B_{d_1}$, such that $P(\{v_1\}, d_1) \neq \bot$.
In case of a tie, we can choose any node.
We claim that $M_1 = P(\{v_1\}, d_1)$ is a \mec.
We repeat this procedure, without taking into account vertices from $M_1$.
This process is continued, as long as a feasible node and vertex can be found.
We now show that it is correct.

\begin{lemma}
For each node $d$ and $v \in B_d$, if $P(\{v\}, d) \neq \bot$, then $P(\{v\}, d)$ is an end-component of $G$.
\end{lemma}

\begin{proof}
From the definition of $P(\cdot, \cdot)$, we have that for every $u \in P(\{v\}, d) \cap V_P$ and every $ux \in E$, it holds that $x \in P(\{v\}, d)$.
Moreover, from each vertex of $P(\{v\}, d)$ there is a path to $v$ and from $v$ there is a path to each vertex of $P(\{v\}, d)$.
It follows that there is a path between any pair of vertices in $P(\{v\}, d)$, so it is a strongly connected set in $G$, thus also an end-component.
\qed
\end{proof}

This implies that our algorithm finds a collection of end-components.
We now show that each such end-component is a MEC.
Let $M$ be an arbitrary \mec and let $d$ be the rootmost node, such that $B_d \cap M \neq \emptyset$.
Since the tree decomposition is nice, $B_d \cap M$ contains a single vertex $v$.
From Lemma~\ref{lemma:mec_intersection} it follows that $M = P(\{v\}, d)$.
It is easy to see that when the algorithm picks a first vertex from $M$, it picks the vertex $v$ defined above, and thus finds a \mec $M$.
It follows easily that every \mec is eventually found by the algorithm.

Let us now discuss the running time.
As before, the dynamic programming step requires $O(n\cdot 2^k\cdot k^{\kexp})$ time.
Retrieving all MECs from their implicit representations requires time that is bounded by the total time of building these representations.
Moreover, the process of finding rootmost nodes requires time that is linear in the size of the tree decomposition.
Hence, the running time is bounded by the time of the dynamic programming and amounts to $O(n\cdot 2^k\cdot k^{\kexp})$.

\begin{theorem}
Given an MDP and the tree decomposition of width $k$ of the MDP graph, the \mec decomposition can be computed in 
$O(n\cdot 2^k\cdot k^{\kexp})$ time, where $n$ is the number of states (vertices).
\end{theorem}

\section{Static and Decremental Algorithms for \mec decomposition and Almost-sure Reachability}
\label{sec:BasicAlgorithms}
In this section we will present the $O(m \cdot k \cdot \log n)$-time static 
algorithms for the \mec decomposition and the \asr set computation,
and the decremental algorithms.
The key would be to present two simple algorithms for the problems that we will view as
decremental graph algorithmic problems (decremental scc computation for \mec decomposition,
and decremental directed reachability for \asr computation).
We will then use dynamic graph algorithmic techniques to obtain the desired result.
We start with the two basic algorithms.
The most straightforward implementations of both these algorithms are not efficient, but we later show that they can be speeded up significantly 
for graphs with low treewidth using dynamic graph algorithmic techniques.

\subsection{Basic algorithms}\label{subsec:basic}

\subsubsection{\mec decomposition.}
We first give an algorithm (formal description as 
Algorithm~\ref{algorithm-mec}) for computing \mec decomposition.
Here, $\textsc{ComputeSccs}$ denotes a function, which computes an array 
$SCC$ that maps the vertices $v$ 
into unique identifiers $SCC[v]$ of the strongly connected components in the graph.

\begin{algorithm}
\caption{\textsc{Mec}(G)}
\label{algorithm-mec}
\begin{algorithmic}[1]
\State $G' := G$
\State $SCC := \textsc{ComputeSccs}(G')$
\While{$\exists_{u \in V_P \cap V(G')} \exists_{uv \in E(G)} SCC[u] \neq SCC[v]$}
	\State{remove $u$ from $G'$}
	\State $SCC := \textsc{ComputeSccs}(G')$
\EndWhile
\end{algorithmic}
\end{algorithm}

\begin{lemma}
\label{lemma:algorithm-mec-correct}
Algorithm \ref{algorithm-mec} is correct.
\end{lemma}

\begin{proof}
The algorithm removes a subset of vertices of $G$, thus obtaining a graph $G'$.
It follows clearly that once the algorithm terminates, the strongly connected components of $G'$ form a \mec decomposition of $G'$.
Moreover, they are end-components in $G$ (note that we use $E(G)$ instead of $E(G')$ in the condition in the third line).
To show that these sets form a \mec decomposition for $G$ (i.e., they are maximal with respect to inclusion), we prove that every vertex $u$ that is removed does not belong to any \mec of $G$.
If $u$ belongs to some \mec $M$, then $v$ must also belong to $M$.
But, by the definition of a strongly connected component, $u$ is not reachable from $v$, so they cannot belong to the same MEC.
Hence, $u$ is not contained in any MEC.
\qed
\end{proof}

\subsubsection{Almost-sure reachability.}
\begin{algorithm}
\caption{\textsc{Asr}(G, s)}
\label{algorithm-asr}
\begin{algorithmic}[1]
\State $G' := G$
\State $A := \textsc{FindReachable}(G', s)$

\While{$\exists_{u \in V_P \cap A} \exists_{uv \in E(G)} v \not\in A$}
	\State{remove $u$ from $G'$}
	\State $A := \textsc{FindReachable}(G', s)$
\EndWhile
\end{algorithmic}
\end{algorithm}
A similar algorithm to the one above can be given for ASR.
Procedure $\textsc{FindReachable}$ computes the set of vertices that are connected to $s$ with a path in $G$.
The formal description is given as Algorithm~\ref{algorithm-asr}.

\begin{lemma}
\label{lemma:algorithm-asr-correct}
Algorithm~\ref{algorithm-asr} is correct.
\end{lemma}

\begin{proof}
The algorithm removes a subset of vertices of $G$, thus obtaining a graph $G'$.
It follows clearly that once the algorithm terminates, the set of vertices from which there is a path to $s$ is an \asr set in $G'$ that satisfies both global and local conditions.
To show that it is also an \asr set in $G$ (i.e. it is maximal with respect to inclusion), we prove that every vertex $u$ that is removed cannot belong to the \asr set.
If $u$ belonged to the set, then $v$ would also belong to it.
But there is no path from $v$ to $s$ in $G$, so $v$ cannot belong to the \asr set, and neither can $u$.
\qed
\end{proof}


\subsection{Static algorithms for \mec and \asr}
This section describes efficient implementations of algorithms from Section~\ref{subsec:basic} that work for graphs with 
low treewidth.

\smallskip\noindent{\bf \mec decomposition.}
In order to compute \mec decomposition, we need to give an efficient implementation of Algorithm~\ref{algorithm-mec}.
This consists in maintaining the array $SCC$ under a sequence of vertex deletions.
Note that instead of removing vertices, we might well just remove all its incident edges.

To maintain strongly connected components we use a data structure by Łącki~\cite{Lacki11}.
Given the tree decomposition of a graph of width $k$, it can maintain the $SCC$ array subject to edge deletions.
The total running time of all delete operations is $O(m\cdot k \cdot \log n)$, and every query to the array is answered in constant time.
Thus, if $\Omega(m)$ edges are deleted, the amortized time of one update is $O(k \cdot \log n)$.

After each update, if a strongly connected component decomposes into multiple strongly connected components, some edges that used to be contained in a single strongly connected component now connect different strongly connected components.
It is easy to see that it suffices to check the condition from the third line of the algorithm just for these edges.
The algorithm maintaining strongly connected components can be easily extended to report the desired edges with no additional overhead.
This way, we obtain an algorithm that computes the \mec decomposition in $O(m\cdot k \cdot \log n)$ total time.

\smallskip\noindent{\bf Almost-sure reachability.}
We now describe an efficient implementation of Algorithm~\ref{algorithm-asr}.
This time it suffices to give an efficient algorithm that maintains the subset $A \subseteq V$ of vertices, such that for every $r \in A$ there exists an $r$-to-$s$ path in $G$.
After reversing all edges in the graph this becomes a single-source reachability problem.
We show that by modifying the algorithm of Łącki~\cite{Lacki11}, this can be achieved in $O(k \cdot \log n)$ amortized time.
We describe the details of the algorithm below.

\smallskip\noindent{\em Decremental single-source reachability.}
Given a directed graph $G$ with a designated \emph{source} $s \in V(G)$, the goal is to maintain the set of vertices reachable from $s$ when the edges of $G$ are deleted.
Moreover, we assume that we are given the tree decomposition of $G$ of width $k$.

The algorithm is a simplified version of the algorithm for decremental all-pairs reachability by Łącki~\cite{Lacki11}.
The description in~\cite{Lacki11} contains an error in the running time analysis of the all-pairs reachability.
However, the problem disappears, if there is only a single source.

One of the ingredients of the algorithm is an algorithm for decremental single-source reachability in a DAG.
The algorithm is very simple.
In the beginning we delete all vertices that are not reachable from the source.
Then, after an edge is deleted, we delete vertices (different from $s$) whose in-degree is $0$, until all remaining vertices have positive in-degree.
Note that deleting a vertex might decrease the in-degree of other vertices and trigger further deletions.
The correctness of the algorithm follows easily.
Moreover, it can be implemented, so that the total running time is linear in the number of edges of the initial graph.
This is because every edge is examined when its start vertex is deleted and this means that the edge itself also gets deleted.

We can now proceed to the algorithm dealing with the general case.
It maintains the subgraph of the initial graph that is reachable from $s$.
In the description we treat $G$ as a variable denoting this subgraph.
To represent $G$ we store its \emph{condensation} $G_c$, that is the graph obtained from $G$ by contracting all strongly connected components.
It is easy to see that a condensation of an arbitrary graph is acyclic.
Hence, we can use the algorithm given above to maintain it.
On the other hand, to maintain the strongly connected components of $G$, we use the data structure by Łącki~\cite{Lacki11}.

When an edge belonging to the condensation is deleted, we can simply update the condensation DAG, deleting some vertices, if necessary.
All other edges are contained inside strongly connected components, so the deletion is handled by the data structure.
This might cause some strongly connected component to break.
In such case the data structure can report the condensation of the subgraph obtained from breaking the component with no additional overhead.
This subgraph is then planted in place of the appropriate vertex in the condensation.
The details are given in~\cite{Lacki11}.

The total running time of processing all edge deletions is $O(m \cdot k \cdot \log n)$ and the set of reachable vertices is maintained explicitly.
Also recall that for treewidth $k$ we have $m=O(n \cdot k)$.

\begin{theorem}
Given an MDP and its tree decomposition of width $k$, the \mec decomposition and the \asr set can be computed in time 
$O(m\cdot k\cdot \log n)$, where $n$ is the number of states (vertices) and $m$ is the number of edges.
\end{theorem}

\subsection{Decremental algorithms}
Both algorithms that we have described can be easily extended to decremental algorithms that support edge deletions.
However, only deleting edges $uv \in E$ such that $u \in V_1$ is allowed.
This assures that the \asr set can only shrink and that every end-component in the \mec decomposition is a subset of a some end-component from the graph before the deletion.

\smallskip\noindent{\bf Almost-sure reachability.}
The algorithm first runs Algorithm~\ref{algorithm-asr} during the initialization phase and computes the initial set $A$.
The set $A$ is maintained by a single-source decremental reachability algorithm.
The very same high-level algorithm can be used to update the set $A$ after an edge is deleted.
We run this algorithm whenever an edge is deleted.
Observe that if we detect that $A$ shrinks, i.e. a subset $U \subseteq A$ of vertices is removed from $A$, we need to check the condition in the third line only for edges that are entering this set.
Thus, each edge is inspected at most once during the entire course of the algorithm.
Hence, the dominating operation is the running time of the decremental single-source reachability algorithm, which requires $O(m \cdot k\cdot \log n)$ time over all deletions or 
$O(k \cdot \log n)$ amortized time for a single deletion, if $\Omega(m)$ edges are deleted.
The proof of correctness is analogous to the one in Lemma~\ref{lemma:algorithm-asr-correct}.

\smallskip\noindent{\bf \mec decomposition.}
We use the same idea as for the decremental algorithm for the \asr set.
In this case Algorithm~\ref{algorithm-mec} can be used both for the initialization and after an edge is deleted.
By maintaining the array $SCC$ with a data structure for decremental SCC maintenance, we get that the amortized time of processing a single update is $O(k \cdot \log n)$.

\begin{theorem}
Given an MDP and its tree decomposition of width $k$, the \mec decomposition and the \asr set can be computed under 
the deletion of $\Omega(m)$ player-1 edges, in amortized time $O(k\cdot \log n)$ per edge deletion,
where $n$ is the number of states (vertices) and $m$ is the number of edges.
\end{theorem}

\noindent{\em Concluding remarks.} 
In this work, we presented faster static and decremental algorithms for 
two core algorithmic problems for MDPs when the treewidth is low.
An interesting question for future work is whether the algorithms
can be extended to MDPs with low DAG-width (as done for parity 
games in~\cite{BDHKO12}).

{\scriptsize 
\smallskip\noindent{\bf Acknowledgements.} The authors would like to thank Monika 
Henzinger for several interesting discussions on related topics.
The research was supported by FWF Grant No P 23499-N23,  
FWF NFN Grant No S11407-N23 (RiSE), ERC Start grant (279307: Graph Games), 
and Microsoft faculty fellows award.
Jakub Łącki is a recipient of the Google Europe Fellowship in Graph Algorithms, and this research is supported in part by this Google Fellowship.
}

\clearpage

\appendix
\section*{Appendix}

\section{Proof of Lemma~\ref{lemm:equalA}}
We prove inclusion in both directions:
\begin{itemize}
\item Every vertex in $A$ must have a path with vertices in $A$ to $s$ in the 
graph: to ensure almost-sure reachability, simple graph reachability must be 
ensured, and the almost-sure set should never be left.
Thus the global condition is satisfied by $A$. 
Since from vertices outside $A$, almost-sure reachability cannot be ensured,
vertices $u \in A \cap V_P$ must have their out-going edges in $A$, as otherwise
the set $A$ is left with positive probability and from the remaining vertices
almost-sure reachability cannot be ensured. 
It follows that $A$ satisfies both the global and the local conditions, and since $A^*$ is
the maximum such set we have $A \subseteq A^*$.

\item We now argue that from every vertex in $A^*$, almost-sure reachability to $s$ is ensured.
By the global condition, every vertex in $A^*$ has a path to $s$ consisting of vertices
in $A^*$, and thus have an edge to a vertex in $A^*$ that is closer 
(in terms of shortest path) to $s$.  
From every vertex in $A^* \cap V_1$ choose the first edge on the shortest path (inside $A^*$) to $s$. 
Consider the resulting Markov chain obtained for the
set $A^*$ of vertices: then the vertex $s$ is the only recurrent state and thus reached with
probability~1. Hence $A^* \subseteq A$.
\end{itemize}

\section{Proof of Lemma~\ref{lemma:p_is_vps}}
The proof proceeds by the induction on the depth of the subtree rooted at $d$.
We consider each node type separately.
The first type corresponds to the basis of the induction.
Since we assume that $P(B', d) \neq \bot$ and $\bot$ is propagating, we immediately have that all values of $P(\cdot, \cdot)$ we refer to are not equal to $\bot$.

\begin{itemize}
\item \textbf{Leaf} The claim follows trivially.
\item \textbf{Join} $P(B', d)$ is a valid partial solution, as it is a sum of two valid partial solutions.
To show that the transitive closure is computed correctly, we show that every path connecting vertices of $B'$ and going through $P(B', d)$ can be traced in $TC(B', c_1) \cup TC(B', c_2)$.
Consider two vertices of $u, v \in B'$ that are connected with a directed path $\rho$ in $P(B', d)$.
Let us split $\rho$ into subpaths by cutting it in each vertex contained in $B'$.
Denote the resulting subpaths $\rho_1, \ldots, \rho_k$.
We want to show that each such subpath can be traced in $TC(B', c_1)$ or $TC(B', c_2)$, which means that it is either contained in $G_{c_1}[P(B', c_1)]$ or $G_{c_2}[P(B', c_2)]$.

Consider a subpath $\rho_i$.
If it is contained within $B'$, it is also contained both in $G_{c_1}[P(B', c_1)]$ and $G_{c_2}[P(B', c_2)]$.
Otherwise, consider the first vertex $v$ outside $B'$ and w.l.o.g. assume that it belongs to $G_{c_1}[P(B', c_1)]$.
By Lemma~\ref{lemma:separator}, any path from $v$ to $G_{c_2}[P(B', c_2)] \setminus B'$ has to go through $B'$.
But, by definition, $\rho_i$ ends at the first vertex of $B'$ encountered after $v$.
Thus, $\rho_i$ is contained in $G_{c_1}[P(B', c_1)]$.

This shows that $TC(B', d) \subseteq (TC(B', c_1) \cup TC(B', c_2))^{*}$.
The reverse inclusion follows immediately.

\item \textbf{Introduce} Fix $B'$ and $d$.
Let us first verify condition~\ref{vps:second} of Definition~\ref{definition-partial-solution}.
We have that $P(B', d)$ is either $P(B', c)$ or $P(B' \setminus \{w\}, c) \cup \{w\}$.
For each vertex $v \neq w$ there exists a path from $v$ to a vertex in $B'$, as it existed in $P(B', c)$.
Additionally, if $w \in P(B', d)$, the condition holds trivially for $w$.

We now check whether condition~\ref{vps:first} of Definition~\ref{definition-partial-solution} holds.
As $P(B', c)$ is a valid partial solution, the condition can only be violated for out-edges of $w$.
But $w \in B_d$, so if the condition is not satisfied, then $B'$ is not a valid subset of $B_d$.
\item \textbf{Forget} In the second case, that is when $P(B', d) = P(B', c)$, the claim follows easily.
Let us now assume that we set $P(B', d) = P(B' \cup \{w\}, c)$.
We only need to verify condition~\ref{vps:second} of Definition~\ref{definition-partial-solution} 
(condition~\ref{vps:first} follows trivially).

We know that for each $v \in P(B', d)$ there exists a path from $v$ to $B' \cup \{w\}$ that is contained within $P(B', d)$.
Moreover, we have checked that there exists a path from $w$ to a vertex in $B'$.
This means that for each $v \in P(B', d)$ there exists a path from $v$ to $B'$.

It follows immediately that $TC(B', d)$ is computed correctly.
\end{itemize}
The desired result follows.
\qed


\end{document}